\documentclass[prl,floatfix,amsmath,superscriptaddress,twocolumn]{revtex4}
\usepackage{amssymb}
\usepackage{graphicx}
\usepackage{graphics}
\usepackage{amsmath}
\usepackage{amsthm}
\usepackage{color}

\newcommand{\bea}{\begin{eqnarray}}
\newcommand{\eea}{\end{eqnarray}}

\def\bi{\begin{itemize}}
\def\ei{\end{itemize}}
\def\bc{\begin{center}}
\def\ec{\end{center}}

\def\C{\hbox{$\mit I$\kern-.7em$\mit C$}}
\def\R{\hbox{$\mit I$\kern-.6em$\mit R$}}

\def\tr{\mathrm{tr}}

\newtheorem{theorem}{Theorem}
\newtheorem{corollary}[theorem]{Corollary}

\begin{document}

\title{Maximally entangled mixed states for a fixed spectrum do not always exist}
\author{Julio I. de Vicente}\email{jdvicent@math.uc3m.es}
\affiliation{Departamento de Matem\'aticas, Universidad Carlos III de
Madrid, E-28911, Legan\'es (Madrid), Spain}
\affiliation{Instituto de Ciencias Matem\'aticas (ICMAT), E-28049 Madrid, Spain}

\begin{abstract}
Entanglement is a resource under local operations assisted by classical communication (LOCC). Given a set of states $S$, if there is one state in $S$ that can be transformed by LOCC into all other states in $S$, then this state is maximally entangled in $S$. It is a well-known result that the $d$-dimensional Bell state is the maximally entangled state in the set of all bipartite states of local dimension $d$. Since in practical applications noise renders every state mixed, it is interesting to study whether sets of mixed states of relevance enable the notion of a maximally entangled state. A natural choice is the set of all states with the same spectrum. In fact, for any given spectrum distribution on two-qubit states, previous work has shown that several entanglement measures are all maximized by one particular state in this set. This has led to consider the possibility that this family of states could be the maximally entangled states in the set of all states with the same spectrum, which should then maximize \emph{all} entanglement measures. In this work I answer this question in the negative: there are no maximally entangled states for a fixed spectrum in general, i.e.\ for every possible choice of the spectrum. In order to do so, I consider the case of rank-2 states and show that for particular values of the eigenvalues there exists no state that can be transformed to all other isospectral states not only under LOCC but also under the larger class of non-entangling operations. This in particular implies that in these cases the state that maximizes a given entanglement measure among all states with the same spectrum depends on the choice of entanglement measure, i.e.\ it cannot be that the aforementioned family of states maximizes all entanglement measures.
\end{abstract}

\maketitle

%\paragraph{Introduction.}

Entanglement is a strong form of correlation for quantum bipartite systems that has no analogue in classical physics. It plays a key role in the foundations of quantum mechanics and in quantum information theory, and it fuels quantum technologies to overcome the limitations of their classical counterparts. Thus, extensive work has been devoted over the last decades to characterize, classify and quantify the different forms in which this phenomenon manifests itself and to obtain protocols and understand the ultimate limitations for the manipulation of this resource (see e.g.\ the review articles \cite{review1,review2,review3}). All these efforts have given rise to the field of entanglement theory. One of the cornerstone results of this theory is the existence of a maximally entangled state, which provides the most useful form of entanglement and sets a gold standard to quantify this resource.

In more detail, entanglement theory is formulated over the paradigm of local operations and classical communication (LOCC). This is the most general form of quantum dynamics that can be implemented by spatially separated parties that do not exchange quantum communication. Thus, entanglement cannot be created by LOCC alone and this sets the basic rule to order the set of entangled states and to construct entanglement measures. If $\rho$ and $\sigma$ are bipartite quantum states and there exists an LOCC protocol $\Lambda$ such that $\Lambda(\rho)=\sigma$, then $\sigma$ cannot be more entangled than $\rho$ and every entanglement measure $E$ must satisfy that $E(\rho)\geq E(\sigma)$. In particular, for any non-entangled (i.e.\ separable) state $\rho$ and LOCC transformation $\Lambda$, $\Lambda(\rho)$ must remain so and it must hold that $E(\rho)=0$. Local unitary (LU) transformations are special examples of LOCC protocols that are moreover reversible. Therefore, the above entails that every entanglement measure $E$ fulfills that $E(\rho)=E(U_A\otimes U_B \rho U_A^\dag\otimes U_B^\dag)$ for arbitrary bipartite states $\rho$ and unitaries $U_A$ and $U_B$. Hence, LU-equivalent states have the same entanglement properties and, whenever we speak about state transformations in the following, a state should be understood as an equivalence class in the above sense. This formalism has enabled the derivation of a huge collection of different entanglement measures, many of them endowed with meaningful operational interpretations \cite{review1,review2,review3}. Moreover, its simplicity and conceptual elegance has been generalized giving rise to quantum resource theories \cite{resource}, which allow to study in a unified manner other forms of quantum advantage such as Bell nonlocality, coherence, quantum thermodynamics or stabilizer computation, to name a few.

However, unfortunately, the mathematical structure of the set of LOCC transformations is rather cumbersome \cite{locc} and the characterization of LOCC convertibility is in general a formidable problem. In this respect, Nielsen's theorem stands out as a landmark achievement providing a simple characterization of LOCC convertibility between pure but otherwise arbitrary bipartite states \cite{nielsen}. Among other important consequences, this result implies that the $d$-dimensional Bell state,
\begin{equation}
|\phi_d\rangle=\frac{1}{\sqrt{d}}\sum_{i=1}^d|ii\rangle,
\end{equation}
can be transformed by LOCC into any other pure bipartite state of the same dimensions. Since the set of LOCC maps is convex, this entails that it can be transformed as well to any state (i.e.\ including mixed states). Thus, $|\phi_d\rangle$ is the maximally entangled two-qudit state. This means that this state maximizes all entanglement measures among all such states and that it has to be the most useful state for any task to be implemented under the LOCC constraint independently of the particular goal to be accomplished.

Thus, in the LOCC scenario the generalized Bell state is the best resource one can aim for. Nevertheless, the starting point of this work is that in many situations one may not be able to prepare any state of choice. Suppose we are bound to a particular set of states $S$ (which does not include the maximally entangled state), can we define a notion of maximally entangled state in $S$? From the perspective of entanglement theory outlined above the answer is clear: given the set $S$, a state in $S$ is the maximally entangled state in $S$ if it can be transformed by LOCC into any other state in $S$. It is not clear in principle, however, which sets $S$ enable such a notion. Since in practice noise renders all states mixed and pure states are an idealization, a particularly relevant instance is the case where $S$ is a set of mixed states with a predefined mixedness structure. A particularly well-motivated choice is to consider the set of all states with the same spectrum as this characterizes all states that can be prepared by unitary evolution starting from any given noisy separable state \cite{unitary1,unitary2}. In fact, Nielsen's theorem corresponds to the case in which the spectrum is a delta distribution (i.e.\ rank equal to one), so this question is a natural generalization of the above. However, although Nielsen's work has been extended to characterize probabilistic \cite{prob} and approximate \cite{approx} transformations among pure states and transformations from pure states to ensembles \cite{ensemble} and mixed states \cite{mixed}, the case of transformations from and to mixed states seems to have remained unexplored territory.

Remarkably, following previous work in \cite{ishizaka}, Ref.~\cite{mems1} has proven that in the case of two-qubit states for any choice of the spectrum the same state maximizes three different and widely used entanglement measures: the entanglement of formation \cite{eof}, the relative entropy of entanglement \cite{ree} and the negativity \cite{negativity}. This suggests that this family of states could be the maximally entangled two-qubit states for any fixed spectrum. If this was the case, then they should maximize all entanglement measures within the corresponding set of states. %Thus, if some other entanglement measure leads to a different maximizer, this would automatically imply the nonexistence of a maximally entangled state in the set.
However, the evaluation of entanglement measures for mixed states boils down to very hard optimization problems and even for two-qubit states explicit formulae are in general unavailable (with only a few notable exceptions such as the negativity and the entanglement of formation \cite{eof2}). Thus, no advancement has been made in this problem in the last years. In this Letter I will solve this question on the negative: there exists no maximally entangled states for a fixed spectrum in general. In order to prove this I will show that in the case where the rank is 2 there exist choices for the spectrum such that the states of \cite{mems1} cannot be transformed by LOCC into other isospectral states. Since the result of \cite{mems1} and the monotonicity of entanglement measures imply that these states cannot be obtained by LOCC from any other isospectral state either, this leads to the desired result. With this, one can construct an entanglement monotone that is not maximized by the states of \cite{mems1} among all states with the same spectrum. This solves Problem 5 in the Open Quantum Problems List maintained by the Institute for Quantum Optics and Quantum Information (IQOQI) in Vienna \cite{open1,open2}. Thus, this result shows that, unlike in the case of pure states, in the LOCC paradigm there is not a mixed state that is the most useful for all entanglement-related applications among all isospectral states and the best state one could prepare by unitary evolution on a noisy separable input is task-dependent (i.e.\ it depends on the entanglement measure one would like to maximize). Actually, I will prove something stronger. I will not only show that the above conversions are impossible by LOCC but also under the more general class of non-entangling (NE) transformations. Since this is the largest class of transformations that leads to a well-defined resource theory of entanglement in the one-shot scenario \cite{resource}, we therefore have that no such theory allows in general for maximally entangled states for a fixed spectrum. Interestingly, once the problem is approached from this perspective the proof is simple and relies only on elementary techniques. The idea of using outer approximations to LOCC in the study of entanglement theory is old and can be traced back at least to the work of Rains \cite{Rains1,Rains2,Rains3,Rains4}, which considers the classes of separable (SEP) and positive-partial-transpose (PPT) operations. For the particular case of NE transformations see \cite{harrownielsen,brandaoplenio1,brandaoplenio2,brandaodatta,patricia,beyondlocc,lamiregula1,brandaoplenio3,lamiregula2}.

In order to establish the aforementioned result, I begin by setting some notation. Throughout this Letter I will only consider two-qubit states and $\mathcal{D}$ and $\mathcal{S}$ will denote respectively the corresponding sets of states (density matrices) and separable states. That is, the elements of $\mathcal{D}$ are given by $4\times4$ unit-trace Hermitian positive semidefinite matrices and the elements of $\mathcal{S}$ by density matrices $\rho\in\mathcal{D}$ such that
\begin{equation}
\rho=\sum_ip_i|\psi_i\rangle\langle\psi_i|\otimes|\chi_i\rangle\langle\chi_i|,
\end{equation}
for some choice of convex weights $\{p_i\}$ and unit-norm vectors $|\psi_i\rangle,|\chi_i\rangle\in\mathbb{C}^2$ $\forall i$. The precise definition of the class of LOCC maps will not be needed here (the interested reader is referred to e.g.\ \cite{locc,locc2}). I will only use that this class is a subset of the class of NE maps. A completely positive and trace-preserving (CPTP) map $\Lambda:\mathcal{D}\to\mathcal{D}$ is NE if $\Lambda(\rho)\in\mathcal{S}$ $\forall\rho\in\mathcal{S}$. The four Bell states that are all LU-equivalent to the maximally entangled state $|\phi_2\rangle$ and give rise to the Bell basis are denoted by
\begingroup
\allowdisplaybreaks
\begin{align}\label{Bellvector}
|\Phi_1\rangle=\frac{1}{\sqrt{2}}(|00\rangle+|11\rangle)&,\quad |\Phi_2\rangle=\frac{1}{\sqrt{2}}(|00\rangle-|11\rangle),\nonumber\\
|\Phi_3\rangle=\frac{1}{\sqrt{2}}(|10\rangle+|01\rangle)&,\quad |\Phi_4\rangle=\frac{1}{\sqrt{2}}(|10\rangle-|01\rangle),
\end{align}%
\endgroup
%\begin{align}\label{Bellvector}
%|\Phi_1\rangle=\frac{1}{\sqrt{2}}(|00\rangle+|11\rangle)&,\quad |\Phi_2\rangle=\frac{1}{\sqrt{2}}(|00\rangle-|11\rangle),\nonumber\\
%|\Phi_3\rangle=\frac{1}{\sqrt{2}}(|10\rangle+|01\rangle)&,\quad |\Phi_4\rangle=\frac{1}{\sqrt{2}}(|10\rangle-|01\rangle),
%\end{align}
and the corresponding density matrices by~$\Phi_i=|\Phi_i\rangle\langle\Phi_i|$ ($i\in\{1,2,3,4\}$).
%\begin{equation}\label{Bellmatrix}
%\phi_\pm=|\phi_\pm\rangle\langle\phi_\pm|,\quad \psi_\pm=|\psi_\pm\rangle\langle\psi_\pm|.
%\end{equation}
A state $\rho\in\mathcal{D}$ is Bell-diagonal if it is diagonal in the Bell basis, i.e.\
\begin{equation}\label{Belldiagonal}
\rho=\sum_{i=1}^4p_i\Phi_i,
\end{equation}
where $\sum_ip_i=1$ and $p_i\geq0$ $\forall i$. These states are known to be entangled if and only if (iff) $\max_ip_i>1/2$ \cite{belldiagonalstates}. In fact, given any state $\rho\in\mathcal{D}$, its fully entangled fraction is defined as
\begin{equation}\label{entangledfraction}
F(\rho)=\max_{U_A,U_B}\tr(\rho U_A\otimes U_B\Phi_1U_A^\dag\otimes U_B^\dag),
\end{equation}
where the maximization is over arbitrary $2\times2$ unitary matrices $U_A$ and $U_B$, and if $F(\rho)>1/2$ then $\rho$ must be entangled \cite{isotropic}. Thus, in particular, $\tr(\rho\Phi_i)\leq1/2$ must hold for any $i\in\{1,2,3,4\}$ if $\rho$ is separable. For any given spectrum distribution on $\mathcal{D}$, $\vec{\lambda}=(\lambda_1,\lambda_2,\lambda_3,\lambda_4)$ arranged in non-increasing order and with $\sum_i\lambda_i=1$ and $\lambda_i\geq0$ $\forall i$, it was found in \cite{mems1} that the unique maximizer of the entanglement of formation, the relative entropy of entanglement and the negativity is given by (the LU-equivalence class of) the state
\begin{equation}\label{mems}
\rho_{\vec{\lambda}}=\lambda_1\Phi_1+\lambda_2|01\rangle\langle01|+\lambda_3\Phi_2+\lambda_4|10\rangle\langle10|.
\end{equation}
In the following we consider rank-2 instances of the states~(\ref{mems}) and (\ref{Belldiagonal}) denoted by
\begin{equation}
\rho_\lambda=\lambda\Phi_1+(1-\lambda)|01\rangle\langle01|,\quad \sigma_\lambda=\lambda\Phi_1+(1-\lambda)\Phi_3,
\end{equation}
where $\lambda\in[1/2,1]$. Notice that $\rho_\lambda$ and $\sigma_\lambda$ are isospectral for any choice of $\lambda$.

With this, I can now state and prove the main result of this Letter.
\begin{theorem}
There is no NE map $\Lambda$ such that $\Lambda(\rho_\lambda)=\sigma_\lambda$ for $\lambda\in(2/3,1)$.
\end{theorem}
\begin{proof}
We will assume that there exists a NE map $\Lambda$ such that $\Lambda(\rho_\lambda)=\sigma_\lambda$ and we will reach a contradiction in the case that $\lambda\in(2/3,1)$. By the linearity of $\Lambda$ we have that
\begin{equation}\label{Lambdalin}
\lambda\Lambda(\Phi_1)+(1-\lambda)\Lambda(|01\rangle\langle01|)=\lambda\Phi_1+(1-\lambda)\Phi_3.
\end{equation}
Therefore, taking $\tr(\Phi_1\sigma_\lambda)$ in this expression we obtain that
\begin{align}
\lambda&=\lambda\tr(\Phi_1\Lambda(\Phi_1))+(1-\lambda)\tr(\Phi_1\Lambda(|01\rangle\langle01|))\nonumber\\
&\leq\lambda\tr(\Phi_1\Lambda(\Phi_1))+\frac{(1-\lambda)}{2},\label{cond1}
\end{align}
where in the last line we have used that $\Lambda$ is NE and, hence, $\tr(\Phi_1\Lambda(|01\rangle\langle01|))\leq1/2$. This imposes that
\begin{equation}\label{phi+1}
\tr(\Phi_1\Lambda(\Phi_1))\geq\frac{3\lambda-1}{2\lambda}.
\end{equation}
Analogously,
\begin{align}
1-\lambda&=\lambda\tr(\Phi_3\Lambda(\Phi_1))+(1-\lambda)\tr(\Phi_3\Lambda(|01\rangle\langle01|))\nonumber\\
&\leq\lambda\tr(\Phi_3\Lambda(\Phi_1))+\frac{(1-\lambda)}{2},\label{cond2}
\end{align}
which leads to
\begin{equation}\label{phi+2}
\tr(\Phi_3\Lambda(\Phi_1))\geq\frac{1-\lambda}{2\lambda}.
\end{equation}
However, considering these two conditions together we see that
\begin{equation}
1\geq\tr[(\Phi_1+\Phi_3)\Lambda(\Phi_1)]\geq\frac{3\lambda-1}{2\lambda}+\frac{1-\lambda}{2\lambda}=1,
\end{equation}
which imposes that both Eqs.~(\ref{phi+1}) and~(\ref{phi+2}) must hold with equality. In turn, looking at Eqs.~(\ref{cond1}) and~(\ref{cond2}) and using that $\lambda\neq1$, this entails that
\begin{equation}\label{01}
\tr(\Phi_1\Lambda(|01\rangle\langle01|))=\tr(\Phi_3\Lambda(|01\rangle\langle01|))=\frac{1}{2}.
\end{equation}
With these restrictions we consider now the action of $\Lambda$ on
\begin{equation}\label{tau}
\tau=\frac{1}{2}\Phi_1+\frac{1}{4}\Phi_3+\frac{1}{4}\Phi_4=\frac{1}{2}\Phi_1+\frac{1}{4}|01\rangle\langle01|+\frac{1}{4}|10\rangle\langle10|,
\end{equation}
which is a Bell-diagonal state of the form~(\ref{Belldiagonal}) with $\max_ip_i\leq1/2$ and, consequently, separable. However, using Eqs.~(\ref{phi+1}) (which, as we now know, holds with equality) and~(\ref{01}) we arrive at
\begin{align}
\tr(\Phi_1\Lambda(\tau))&=\frac{3\lambda-1}{4\lambda}+\frac{1}{8}+\tr(\Phi_1\Lambda(|10\rangle\langle10|))\nonumber\\
&\geq\frac{3\lambda-1}{4\lambda}+\frac{1}{8}.
\end{align}
But if $\lambda>2/3$, this implies that $\tr(\Phi_1\Lambda(\tau))>1/2$ and, in turn, that $\Lambda(\tau)$ is entangled. Since $\tau$ is separable we have thus reached a contradiction with the assumption that $\Lambda$ is NE.
\end{proof}

The results announced in the introduction follow now immediately from Theorem 1.
\begin{corollary}
There is no maximally entangled mixed state for a fixed spectrum in general.
\end{corollary}
The reason for this is that Theorem 1 implies that for $\vec{\lambda}=(\lambda,1-\lambda,0,0)$ and $\lambda\in(2/3,1)$ the transformation $\rho_\lambda$ to the isospectral state $\sigma_\lambda$ is impossible by LOCC. On the other hand, the results of \cite{mems1} show that there exists an entanglement measure $E$ such that $E(\rho_\lambda)>E(\rho)$ for any $\rho\in\mathcal{D}$ with spectrum $(\lambda,1-\lambda,0,0)$ that is not LU-equivalent to $\rho_\lambda$. Consequently, by the monotonicity of entanglement measures under LOCC the transformation $\rho\to\rho_\lambda$ is impossible by LOCC for any such $\rho$. Therefore, there exists no state with spectrum $(\lambda,1-\lambda,0,0)$ and $\lambda\in(2/3,1)$ that can be transformed by LOCC into any other state with the same spectrum. Furthermore, Theorem 1 forbids the transformation $\rho_\lambda\to\sigma_\lambda$ under the more general class of NE operations. Analogously, no LU-inequivalent state to the state $\rho_\lambda$ can be transformed to it by NE operations since the above strict inequality holds choosing $E$ to be the relative entropy of entanglement, which is known to be monotonic not only under LOCC transformations but also under NE transformations (see Lemma IV.5 in \cite{brandaoplenio2}). Thus, there exists no maximally entangled mixed state for a fixed spectrum in general even in the most general resource theory of entanglement in which the free operations are relaxed from LOCC to NE.

\begin{corollary}
There exist choices of $\vec{\lambda}$ for which the states $\rho_{\vec{\lambda}}$ in Eq.~(\ref{mems}) do not maximize all entanglement monotones among all states with the same spectrum.
\end{corollary}
Let $p_{LOCC}(\rho\to\sigma)$ and $p_{NE}(\rho\to\sigma)$ denote the optimal probability of converting the state $\rho$ to $\sigma$ by respectively LOCC and NE operations. It has been proven in \cite{vidalmonotones} that, for any fixed state $\sigma$, $p_{LOCC}(\rho\to\sigma)$ is an entanglement monotone as a function of $\rho$. Thus, $E_\lambda(\rho)=p_{LOCC}(\rho\to\sigma_\lambda)$ defines a family of entanglement monotones. The corollary follows because for $\lambda\in(2/3,1)$ while it obviously holds that $E_\lambda(\sigma_\lambda)=1$, we on the other hand have that
\begin{equation}
E_\lambda(\rho_\lambda)\leq p_{NE}(\rho_\lambda\to\sigma_\lambda)<1,
\end{equation}
where in the first inequality we have used that the class of LOCC operations is a subset of the class of NE operations and in the second inequality Theorem 1 and that the class of NE maps is closed. This last property follows from the well-known fact that the set of separable states is closed.

In summary, I have introduced here the problem of the existence of a maximally entangled state within a given set of states and I have studied in detail the particular case of isospectral states. While previous work had shown that the states~(\ref{mems}) maximize several important entanglement measures among all two-qubit states with the same spectrum, this Letter proves that this is not the case for all entanglement measures. In fact, it is demonstrated here that maximally entangled states for a fixed spectrum cannot exist in general in the standard resource theory of entanglement: there are choices of the spectrum where no state can be transformed by LOCC into all other states with the same spectrum. This situation is reminiscent of the case of multipartite pure states, where the analogous property holds \cite{mes}. However, in this latter case it has been shown that a notion of maximally entangled multipartite state can be obtained in more general resource theories where the class of allowed transformations is relaxed from LOCC to a bigger but still meaningful class \cite{patricia}. Unfortunately, this cannot be the case for bipartite mixed states with a fixed spectrum. Even though it is known that NE transformations enable bipartite state conversions that are impossible by LOCC \cite{brandaodatta,beyondlocc,brandaoplenio3,lamiregula2}, the impossibility proof obtained here not only applies to the case of LOCC but also to NE transformations, the largest class of allowed transformations that can be considered in a consistent resource theory of entanglement in the one-shot scenario.

Despite this negative result, several questions are left open for future investigation. The entanglement measure which is not maximized by the states~(\ref{mems}) is constructed ad hoc and it would be interesting to study whether there is a more physically relevant measure with this property. Is the fact that the states~(\ref{mems}) maximize for a given spectrum the entanglement of formation, the relative entropy of entanglement and the negativity an accident or is this a property of all entanglement measures with a certain structure? Also, the impossibility argument obtained here only covers certain particular values of the spectrum, do there exist particular spectrum distributions where a maximally entangled state exists beyond the case of rank-1 states (and the trivial case in which all states in the set are separable)? Are there other operationally or physically motivated choices of the set $S$ for which a maximally entangled mixed state in $S$ might exist? What are in general necessary and/or sufficient conditions for a set $S$ with this property to exist?

\begin{acknowledgments}
I acknowledge financial support from the Spanish Ministerio de Ciencia e Innovaci\'on (grant PID2020-113523GB-I00 and ``Severo Ochoa Programme for Centres of Excellence" grant CEX2023-001347-S funded by MCIN/AEI/10.13039/501100011033) and from Comunidad de Madrid (grant QUITEMAD-CM P2018/TCS-4342).
\end{acknowledgments}


\begin{thebibliography}{widest-label}

\bibitem{review1} M. B. Plenio and S. Virmani, Quant. Inf. Comput. \textbf{7}, 1 (2007).

\bibitem{review2} R. Horodecki, P. Horodecki, M. Horodecki, and K. Horodecki, Rev. Mod. Phys. \textbf{81}, 865 (2009).

\bibitem{review3} C. Eltschka and J. Siewert, J. Phys. A: Math. Theor. \textbf{47}, 424005 (2014).

\bibitem{resource} E. Chitambar and G. Gour, Rev. Mod. Phys. \textbf{91}, 025001 (2019).

\bibitem{locc} E. Chitambar, D. Leung, L. Mancinska, M. Ozols, and A. Winter, Commun. Math. Phys. \textbf{328}, 303 (2014).

\bibitem{nielsen} M. A. Nielsen, Phys. Rev. Lett. \textbf{83}, 436 (1999).

\bibitem{unitary1} W. D\"ur, G. Vidal, J. I. Cirac, N. Linden, and S. Popescu, Phys. Rev. Lett. \textbf{87}, 137901 (2001).

\bibitem{unitary2} B. Kraus and J. I. Cirac, Phys. Rev. A \textbf{63}, 062309 (2001).

\bibitem{prob} G. Vidal, Phys. Rev. Lett. \textbf{83}, 1046 (1999).

\bibitem{approx} G. Vidal, D. Jonathan, and M. A. Nielsen, Phys. Rev. A \textbf{62}, 012304 (2000).

\bibitem{ensemble} D. Jonathan and M. B. Plenio, Phys. Rev. Lett. \textbf{83}, 1455 (1999).

\bibitem{mixed} G. Vidal, Phys. Rev. A \textbf{62}, 062315 (2000).

\bibitem{ishizaka} S. Ishizaka and T. Hiroshima, Phys. Rev. A \textbf{62}, 022310 (2000).

\bibitem{mems1} F. Verstraete, K. Audenaert, and B. De Moor, Phys. Rev. A \textbf{64}, 012316 (2001).

%\bibitem{mems2} T.-C. Wei, K. Nemoto, P. M. Goldbart, P. G. Kwiat, W. J. Munro, and F. Verstraete, Phys. Rev. A \textbf{67}, 022110 (2003).

\bibitem{eof} C. H. Bennett, D. P. DiVincenzo, J. A. Smolin, and W. K. Wootters, Phys. Rev. A \textbf{54}, 3824 (1996).

\bibitem{ree} V. Vedral, M. B. Plenio, M. A. Rippin, and P. L. Knight, Phys. Rev. Lett. \textbf{78}, 2275 (1997).

\bibitem{negativity} G. Vidal and R. F. Werner, Phys. Rev. A \textbf{65}, 032314 (2002).

\bibitem{eof2} W. K. Wootters, Phys. Rev. Lett. \textbf{80}, 2245 (1998).

\bibitem{open1} O. Kr\"uger and R. F. Werner, arXiv:quant-ph/0504166 (2005).

\bibitem{open2} \verb"https://oqp.iqoqi.oeaw.ac.at/".

\bibitem{Rains1} E. M. Rains, arXiv:quant-ph/9707002v3 (1997).

\bibitem{Rains2} E. M. Rains, Phys. Rev. A \textbf{60}, 173 (1999).

\bibitem{Rains3} E. M. Rains, Phys. Rev. A \textbf{60}, 179 (1999); Phys. Rev. A \textbf{63}, 019902(E) (2000).

\bibitem{Rains4} E. M. Rains, IEEE Trans. Inf. Theory \textbf{47}, 2921 (2001).

\bibitem{harrownielsen} A. W. Harrow and M. A. Nielsen, Phys. Rev. A \textbf{68}, 012308 (2003).

\bibitem{brandaoplenio1} F. G. S. L. Brand\~{a}o and M. B. Plenio, Nat. Phys. \textbf{4}, 873 (2008).

\bibitem{brandaoplenio2} F. G. S. L. Brand\~{a}o and M. B. Plenio, Commun. Math. Phys. \textbf{295}, 829 (2010).

\bibitem{brandaodatta} F. G. S. L. Brand\~{a}o and N. Datta, IEEE Trans. Inf. Theory \textbf{57}, 1754 (2011).

\bibitem{patricia} P. Contreras-Tejada, C. Palazuelos, and J. I. de Vicente, Phys. Rev. Lett. \textbf{122}, 120503 (2019).

\bibitem{beyondlocc} E. Chitambar, J. I. de Vicente, M. W. Girard, and G. Gour, J. Math. Phys. \textbf{61}, 042201 (2020).

\bibitem{lamiregula1} L. Lami and B. Regula, Nat. Phys. \textbf{19}, 184 (2023).

\bibitem{brandaoplenio3} M. Berta, F. G. S. L. Brand\~{a}o, G. Gour, L. Lami, M. B. Plenio, B. Regula, and M. Tomamichel, Quantum \textbf{7}, 1103 (2023).

\bibitem{lamiregula2} L. Lami and B. Regula, arXiv:2307.11008 (2023).

\bibitem{locc2} M. Hebenstreit, M. Englbrecht, C. Spee, J. I. de Vicente, and B. Kraus, New J. Phys. \textbf{23}, 033046 (2021).

\bibitem{belldiagonalstates} R. Horodecki and M. Horodecki, Phys. Rev. A \textbf{54}, 1838 (1996).

\bibitem{isotropic} M. Horodecki and P. Horodecki, Phys. Rev. A \textbf{59}, 4206 (1999).

%\bibitem{beyondlocc} E. Chitambar, J. I. de Vicente, M. W. Girard, and G. Gour, J. Math. Phys. \textbf{61}, 042201 (2020).

\bibitem{vidalmonotones} G. Vidal, J. Mod. Opt. \textbf{47}, 355 (2000).

\bibitem{mes} J. I. de Vicente, C. Spee, and B. Kraus, Phys. Rev. Lett. \textbf{111}, 110502 (2013).




\end{thebibliography}
\end{document}